\documentclass[11pt,letterpaper]{article}

\def\showauthornotes{0}
\def\showtableofcontents{0}
\def\showkeys{0}
\def\showdraftbox{0}
\def\showcolorlinks{1}
\def\usemicrotype{1}
\def\showfixme{0}

\usepackage{etex}

\usepackage[l2tabu, orthodox]{nag}

\usepackage{xspace,enumerate}

\usepackage[dvipsnames]{xcolor}

\usepackage[american]{babel}

\usepackage{mathtools}
\usepackage{amsmath,amsfonts}

\usepackage{amsthm}

\newtheorem{theorem}{Theorem}[section]
\newtheorem*{theorem*}{Theorem}

\newtheorem{proposition}[theorem]{Proposition}
\newtheorem*{proposition*}{Proposition}
\newtheorem{lemma}[theorem]{Lemma}
\newtheorem*{lemma*}{Lemma}

\newtheorem*{conjecture*}{Conjecture}

\newtheorem*{fact*}{Fact}

\newtheorem*{hypothesis*}{Hypothesis}

\theoremstyle{definition}
\newtheorem{definition}[theorem]{Definition}

\newtheorem{algorithm}[theorem]{Algorithm}
\newtheorem{problem}[theorem]{Problem}

\theoremstyle{remark}

\newtheorem*{claim*}{Claim}
\newtheorem{remark}[theorem]{Remark}
\newtheorem*{remark*}{Remark}
\newtheorem{observation}[theorem]{Observation}
\newtheorem*{observation*}{Observation}

\usepackage[
letterpaper,
top=1in,
bottom=1in,
left=1in,
right=1in]{geometry}

\ifnum\showkeys=1
\usepackage[color]{showkeys}
\fi

\ifnum\showcolorlinks=1
\usepackage[
pagebackref,
colorlinks=true,
urlcolor=blue,
linkcolor=blue,
citecolor=OliveGreen,
]{hyperref}
\fi

\ifnum\showcolorlinks=0
\usepackage[
pagebackref,
colorlinks=false,
pdfborder={0 0 0}
]{hyperref}
\fi

\usepackage{prettyref}

\newcommand{\savehyperref}[2]{\texorpdfstring{\hyperref[#1]{#2}}{#2}}

\newrefformat{eq}{\savehyperref{#1}{\textup{(\ref*{#1})}}}
\newrefformat{eqn}{\savehyperref{#1}{\textup{(\ref*{#1})}}}
\newrefformat{lem}{\savehyperref{#1}{Lemma~\ref*{#1}}}
\newrefformat{def}{\savehyperref{#1}{Definition~\ref*{#1}}}
\newrefformat{thm}{\savehyperref{#1}{Theorem~\ref*{#1}}}
\newrefformat{cor}{\savehyperref{#1}{Corollary~\ref*{#1}}}
\newrefformat{cha}{\savehyperref{#1}{Chapter~\ref*{#1}}}
\newrefformat{sec}{\savehyperref{#1}{Section~\ref*{#1}}}
\newrefformat{app}{\savehyperref{#1}{Appendix~\ref*{#1}}}
\newrefformat{tab}{\savehyperref{#1}{Table~\ref*{#1}}}
\newrefformat{fig}{\savehyperref{#1}{Figure~\ref*{#1}}}
\newrefformat{hyp}{\savehyperref{#1}{Hypothesis~\ref*{#1}}}
\newrefformat{alg}{\savehyperref{#1}{Algorithm~\ref*{#1}}}
\newrefformat{rem}{\savehyperref{#1}{Remark~\ref*{#1}}}
\newrefformat{item}{\savehyperref{#1}{Item~\ref*{#1}}}
\newrefformat{step}{\savehyperref{#1}{step~\ref*{#1}}}
\newrefformat{conj}{\savehyperref{#1}{Conjecture~\ref*{#1}}}
\newrefformat{fact}{\savehyperref{#1}{Fact~\ref*{#1}}}
\newrefformat{prop}{\savehyperref{#1}{Proposition~\ref*{#1}}}
\newrefformat{prob}{\savehyperref{#1}{Problem~\ref*{#1}}}
\newrefformat{claim}{\savehyperref{#1}{Claim~\ref*{#1}}}
\newrefformat{relax}{\savehyperref{#1}{Relaxation~\ref*{#1}}}
\newrefformat{red}{\savehyperref{#1}{Reduction~\ref*{#1}}}
\newrefformat{part}{\savehyperref{#1}{Part~\ref*{#1}}}

\newcommand{\Sref}[1]{\hyperref[#1]{\S\ref*{#1}}}

\usepackage{nicefrac}

\ifnum\usemicrotype=1
\usepackage{microtype}
\fi

\ifnum\showauthornotes=1
\newcommand{\Authornote}[2]{{\sffamily\small\color{red}{[#1: #2]}}}
\newcommand{\Authornotecolored}[3]{{\sffamily\small\color{#1}{[#2: #3]}}}
\newcommand{\Authorcomment}[2]{{\sffamily\small\color{gray}{[#1: #2]}}}
\newcommand{\Authorstartcomment}[1]{\sffamily\small\color{gray}[#1: }

\newcommand{\Authorfnote}[2]{\footnote{\color{red}{#1: #2}}}
\newcommand{\Authorfixme}[1]{\Authornote{#1}{\textbf{??}}}
\newcommand{\Authormarginmark}[1]{\marginpar{\textcolor{red}{\fbox{\Large #1:!}}}}
\else
\newcommand{\Authornote}[2]{}
\newcommand{\Authornotecolored}[3]{}
\newcommand{\Authorcomment}[2]{}
\newcommand{\Authorstartcomment}[1]{}

\newcommand{\Authorfnote}[2]{}
\newcommand{\Authorfixme}[1]{}
\newcommand{\Authormarginmark}[1]{}
\fi

\ifnum\showfixme=0

\fi

\usepackage{boxedminipage}

\newcommand{\Paren}[1]{\left(#1\right)}

\newcommand{\iprod}[1]{\langle#1\rangle}
\newcommand{\Iprod}[1]{\left\langle#1\right\rangle}

\newcommand{\Esymb}{\mathbb{E}}
\newcommand{\Psymb}{\mathbb{P}}

\DeclareMathOperator*{\E}{\Esymb}

\DeclareMathOperator*{\ProbOp}{\Psymb}

\renewcommand{\Pr}{\ProbOp}

\newcommand{\tensor}{\otimes}

\newcommand{\textparen}[1]{\text{(#1)}}

\ifx\because\undefined
\newcommand{\because}[1]{\textparen{because #1}}
\else
\renewcommand{\because}[1]{\textparen{because #1}}
\fi

\newcommand\bdot\bullet

\ifx\mathds\undefined 

\else
}
\fi

\DeclareMathOperator{\poly}{poly}

\DeclareMathOperator{\Span}{span}
\DeclareMathOperator{\Proj}{proj}

\newcommand{\R}{\mathbb R}

\newcommand{\cD}{\mathcal D}

\newcommand{\cM}{\mathcal M}

\newcommand{\cP}{\mathcal P}

\newcommand{\cT}{\mathcal T}

\ifnum\showdraftbox=1
\newcommand{\draftbox}{\begin{center}
  \fbox{%
    \begin{minipage}{2in}%
      \begin{center}%
          \Large\textsc{Working Draft}\\%
        Please do not distribute%
      \end{center}%
    \end{minipage}%
  }%
\end{center}
\vspace{0.2cm}}
\else
\newcommand{\draftbox}{}
\fi

\let\epsilon=\varepsilon

\numberwithin{equation}{section}

\newcommand\MYcurrentlabel{xxx}

\newcommand{\MYstore}[2]{%
  \global\expandafter \def \csname MYMEMORY #1 \endcsname{#2}%
}

\newcommand{\MYload}[1]{%
  \csname MYMEMORY #1 \endcsname%
}

\newcommand{\MYnewlabel}[1]{%
  \renewcommand\MYcurrentlabel{#1}%
  \MYoldlabel{#1}%
}

\newcommand{\MYdummylabel}[1]{}

\newcommand{\torestate}[1]{%
  \let\MYoldlabel\label%
  \let\label\MYnewlabel%
  #1%
  \MYstore{\MYcurrentlabel}{#1}%
  \let\label\MYoldlabel%
}

\newcommand{\restatetheorem}[1]{%
  \let\MYoldlabel\label
  \let\label\MYdummylabel
  \begin{theorem*}[Restatement of \prettyref{#1}]
    \MYload{#1}
  \end{theorem*}
  \let\label\MYoldlabel
}

\newcommand{\restatelemma}[1]{%
  \let\MYoldlabel\label
  \let\label\MYdummylabel
  \begin{lemma*}[Restatement of \prettyref{#1}]
    \MYload{#1}
  \end{lemma*}
  \let\label\MYoldlabel
}

\newcommand{\restateprop}[1]{%
  \let\MYoldlabel\label
  \let\label\MYdummylabel
  \begin{proposition*}[Restatement of \prettyref{#1}]
    \MYload{#1}
  \end{proposition*}
  \let\label\MYoldlabel
}

\newcommand{\restatefact}[1]{%
  \let\MYoldlabel\label
  \let\label\MYdummylabel
  \begin{fact*}[Restatement of \prettyref{#1}]
    \MYload{#1}
  \end{fact*}
  \let\label\MYoldlabel
}

\newcommand{\restate}[1]{%
  \let\MYoldlabel\label
  \let\label\MYdummylabel
  \MYload{#1}
  \let\label\MYoldlabel
}

\newcommand{\addreferencesection}{
  \phantomsection
  \addcontentsline{toc}{section}{References}
}

\let\origparagraph\paragraph
\renewcommand{\paragraph}[1]{\origparagraph{#1.}}

\allowdisplaybreaks

\usepackage{paralist}

\usepackage{comment}

\let\citet\cite

\theoremstyle{definition}

\DeclareUrlCommand\email{}

\newcommand{\PO}{\mathcal{P}_\Omega}
\newcommand{\pref}{\prettyref}
\newcommand{\Oc}{\overline{\Omega}}
\usepackage{enumitem}
\newcommand{\tO}{{\tilde O}}

\title{Symmetric Tensor Completion from Multilinear Entries and Learning Product Mixtures over the Hypercube}

\author{%
\normalsize
Tselil Schramm\thanks{UC Berkeley, \protect\email{tschramm@cs.berkeley.edu}.
Supported by an NSF Graduate Research Fellowship (NSF award no 1106400).}
\and
\normalsize
Ben Weitz\thanks{UC Berkeley, \protect\email{bsweitz@cs.berkeley.edu}. Supported by an NSF Graduate Research Fellowship (NSF award no DGE 1106400).}
}

\begin{document}

\maketitle

\draftbox

\thispagestyle{empty}

\begin{abstract}
   We give an algorithm for completing an order-$m$ symmetric low-rank tensor from its multilinear entries in time roughly proportional to the number of tensor entries.
We apply our tensor completion algorithm to the problem of learning mixtures of product distributions over the hypercube, obtaining new algorithmic results.
    If the centers of the product distribution are linearly independent, then we recover distributions with as many as $\Omega(n)$ centers in polynomial time and sample complexity.
    In the general case, we recover distributions with as many as $\tilde\Omega(n)$ centers in quasi-polynomial time, answering an open problem of Feldman et al. (SIAM J. Comp.) for the special case of distributions with incoherent bias vectors.

    Our main algorithmic tool is the iterated application of a low-rank matrix completion algorithm for matrices with adversarially missing entries.

\end{abstract}

\clearpage

\ifnum\showtableofcontents=1
{
\tableofcontents
\thispagestyle{empty}
 }
\fi

\clearpage

\setcounter{page}{1}

\section{Introduction}

Suppose we are given sample access to a distribution over the hypercube $\{\pm 1\}^n$, where each sample $x$ is generated in the following manner: there are $k$ product distributions $\cD_1,\ldots, \cD_k$ over $\{\pm 1\}^n$ (the $k$ ``centers'' of the distribution), and $x$ is drawn from $\cD_i$ with probability $p_i$.
This distribution is called a product mixture over the hypercube.

Given such a distribution, our goal is to recover from samples the parameters of the individual product distributions.
That is, we would like to estimate the probability $p_i$ of drawing from each product distribution, and furthermore we would like to estimate the parameters of the product distribution itself.
This problem has been studied extensively and approached with a variety of strategies (see e.g. \cite{FM99,CRao08, FOS08}).

A canonical approach to problems of this type is to empirically estimate the moments of the distribution, from which it may be possible to calculate the distribution parameters using linear-algebraic tools (see e.g. \cite{AM05, MR06,FOS08,AGHkT14}, and many more).
For product distributions over the hypercube, this technique runs into the problem that the square moments are always $1$, and so they provide no information.

The seminal work of Feldman, O'Donnell and Servedio \cite{FOS08} introduces an approach to this problem which compensates for the missing higher-order moment information using matrix completion.
Via a restricted brute-force search, Feldman et al. check all possible square moments, resulting in an algorithm that is triply-exponential in the number of distribution centers.
Continuing this line work, by giving an alternative to the brute-force search Jain and Oh \cite{JO13} recently obtained a polynomial-time algorithm for a restricted class of product mixtures.
In this paper we extend these ideas, giving a polynomial-time algorithm for a wider class of product mixtures, and a quasi-polynomial time algorithm for an even broader class of product mixtures (including product mixtures with centers which are not linearly independent).

Our main tool is a matrix-completion-based algorithm for completing tensors of order $m$ from their multilinear moments in time $\tilde{O}(n^{m+1})$, which we believe may be of independent interest.
There has been ample work in the area of noisy tensor decomposition (and completion), see e.g. \cite{JO14,BKS15,TS15,BM15}.
However, these works usually assume that the tensor is obscured by random noise, while in our setting the ``noise'' is the absence of all non-multilinear entries.
An exception to this is the work of \cite{BKS15}, where to obtain a quasi-polynomial algorithm it suffices to have the injective tensor norm of the noise be bounded via a Sum-of-Squares proof.\footnote{It may be possible that this condition is met for some symmetric tensors when only multilinear entries are known, but we do not know an SOS proof of this fact.}
To our knowledge, our algorithm is the only $n^{O(m)}$-time algorithm that solves the problem of completing a symmetric tensor when only multilinear entries are known.

\subsection{Our Results}

Our main result is an algorithm for learning a large subclass of product mixtures with up to even $\Omega(n)$ centers in polynomial (or quasi-polynomial) time.
The subclass of distributions on which our algorithm succeeds is described by characteristics of the subspace spanned by the bias vectors.
Specifically, the rank and \emph{incoherence} of the span of the bias vectors cannot simultaneously be too large.
Intuitively, the incoherence of a subspace measures how close the subspace is to a coordinate subspace of $\R^n$.
We give a formal definition of incoherence later, in \prettyref{def:incoherence}.

More formally, we prove the following theorem:

\begin{theorem}\label{thm:main_learn_prod}
    Let $\cD$ be a mixture over $k$ product distributions on $\{\pm 1\}^n$, with bias vectors $v_1,\ldots,v_k \in \R^n$ and mixing weights $w_1,\ldots,w_k>0$.
    Let $\Span\{v_i\}$ have dimension $r$ and incoherence $\mu$.
    Suppose we are given as input the moments of $\cD$.
    \begin{enumerate}
    \item If $v_1,\dots,v_k$ are linearly independent, then as long as $4 \cdot \mu \cdot r < n$, there is a $\poly(n,k)$ algorithm that recovers the parameters of $\cD$.
    \item Otherwise, if $\left|\iprod{ v_i, v_j} \right| < \|v_i\|\cdot\|v_j\|\cdot (1-\eta)$ for every $i\neq j$ and $\eta > 0$, then as long as $4\cdot\mu \cdot r\cdot \log k/\log\tfrac{1}{1-\eta} < n$, there is an $n^{O(\log k/\log \frac{1}{1-\eta})}$ time algorithm that recovers the parameters of $\cD$.
    \end{enumerate}
\end{theorem}
\begin{remark}
In the case that $v_1,\ldots,v_k$ are not linearly independent, the runtime depends on the separation between the vectors.
We remark however that if we have some $v_i = v_j$ for $i\neq j$, then the distribution is equivalently representable with fewer centers by taking the center $v_i$ with mixing weight $w_i + w_j$.
If there is some $v_i = - v_j$, then our algorithm can be modified to work in that case as well, again by considering $v_i$ and $v_j$ as one center--we detail this in \pref{sec:pdist}.
\end{remark}

    In the main body of the paper we assume access to exact moments; in \pref{app:error} we prove \pref{thm:learn-big}, a version of \pref{thm:main_learn_prod} which accounts for sampling error.

    The foundation of our algorithm for learning product mixtures is an algorithm for completing a low-rank incoherent tensor of arbitrary order given access only to its multilinear entries:

\begin{theorem}
    \label{thm:tensor-complete-alg-main}
    Let $T$ be a symmetric tensor of order $m$, so that $T = \sum_{i\in [k]}w_i \cdot v_i^{\tensor m}$ for some vectors $v_1,\ldots,v_k \in \R^n$ and scalars $w_1,\ldots, w_k \neq 0$.
    Let $\Span\{v_i\}$ have incoherence $\mu$ and dimension $r$.
    Given perfect access to all multilinear entries of $T$, if $4\cdot\mu\cdot r\cdot m/n < 1$, then there is an algorithm which returns the full tensor $T$ in time $\tilde O(n^{m+1})$.
\end{theorem}

\subsection{Prior Work}
We now discuss in more detail prior work on learning product mixtures over the hypercube, and contextualize our work in terms of previous results.

The pioneering papers on this question gave algorithms for a very restricted setting: the works of \cite{FM99} and \cite{C99,CGG01} introduced the problem and gave algorithms for learning a mixture of exactly two product distributions over the hypercube.

The first general result is the work of Feldman, O'Donnell and Servedio, who give an algorithm for learning a mixture over $k$ product distributions in $n$ dimensions in time $n^{O(k^3)}$ with sample complexity $n^{O(k)}$.
Their algorithm relies on brute-force search to enumerate all possible product mixtures that are consistent with the observed second moments of the distribution.
After this, they use samples to select the hypothesis with the Maximum Likelihood.
Their paper leaves as an open question the more efficient learning of discrete mixtures of product distributions, with a smaller exponential dependence (or even a quasipolynomial dependence) on the number of centers.\footnote{
    We do not expect better than quasipolynomial dependence on the number of centers, as learning the parity distribution on $t$ bits is conjectured to require at least $n^{\Omega(t)}$ time, and this distribution can be realized as a product mixture over $2^{t-1}$ centers.
}

More recently, Jain and Oh \cite{JO13} extended this approach: rather than generate a large number of hypotheses and pick one, they use a tensor power iteration method of \cite{AGHkT14} to find the right decomposition of the second- and third-order moment tensors. To learn these moment tensors in the first place, they use alternating minimization to complete the (block)-diagonal of the second moments matrix, and they compute a least-squares estimation of the third-order moment tensor.
Using these techniques, Jain and Oh were able to obtain a significant improvement for a restricted class of product mixtures, obtaining a polynomial time algorithm for linearly independent mixtures over at most $k = O(n^{2/7})$ centers.
In order to ensure the convergence of their matrix (and tensor) completion subroutine, they introduce constraints on the span of the bias vectors of the distribution
(see \pref{sec:pinc} for a discussion of incoherence assumptions on product mixtures).
Specifically, letting $r$ the rank of the span, letting $\mu$ be the incoherence of the span, and letting $n$ be the dimension of the samples, they require that $\tilde\Omega(\mu^5 r^{7/2}) \le n$.\footnote{
    The conditions are actually more complicated, depending on the condition number of the second-moment matrix of the distribution.
    For precise conditions, see \cite{JO13}.
}
Furthermore, in order to extract the bias vectors from the moment information, they require that the bias vectors be linearly independent.
When these conditions are met by the product mixture, Jain and Oh learn the mixture in polynomial time.

In this paper, we improve upon this result, and can handle as many as $\Omega(n)$ centers in some parameter settings.
Similarly to \cite{JO13}, we use as a subroutine an algorithm for completing low-rank matrices with adversarially missing entries.
However, unlike \cite{JO13}, we use an algorithm with more general guarantees, the algorithm of \cite{HKZ11}.\footnote{
    A previous version of this paper included an analysis of a matrix completion algorithm almost identical to that of \cite{HKZ11}, and claimed to be the first adversarial matrix completion result of this generality. Thanks to the comments of an anonymous reviewer, we were notified of our mistake.
}
These stronger guarantees allow us to devise an algorithm for completing low-rank higher-order tensors from their multilinear entries, and this algorithm
allows us to obtain a polynomial time algorithm for a more general class of linearly independent mixtures of product distributions than \cite{JO13}.

Furthermore, because of the more general nature of this matrix completion algorithm, we can give a new algorithm for completing low-rank tensors of arbitrary order given access only to the multilinear entries of the tensor.
Leveraging our multilinear tensor completion algorithm, we can reduce the case of linearly dependent bias vectors to the linearly independent case by going to higher-dimensional tensors.
This allows us to give a quasipolynomial algorithm for the general case, in which the centers may be linearly dependent.
To our knowledge, \pref{thm:main_learn_prod} is the first quasi-polynomial algorithm that learns product mixtures whose centers are not linearly independent.

\begin{figure}[t]
  \begin{center}
      {\bf Learning Product Mixtures with $k$ Centers over $\{\pm 1\}^n$}
      \renewcommand{\arraystretch}{1.4}
  \begin{tabular}{|c|c|c|c|c|c|}
    \hline
    Reference & Runtime & Samples & Largest $k$ & Dep. Centers?& Incoherence? \\ \hline \hline
    Feldman et al. \cite{FOS08} & $n^{O(k^3)}$ & $n^{O(k)}$ & $n$ & Allowed & Not Required \\ \hline
    Jain \& Oh \cite{JO14} & poly$(n,k)$  & poly$(n,k)$ & $k \le O(n^{2/7})$ &  Not Allowed & Required\\ \hline
     Our Results {\Large $\substack{\text{lin. indep.}\\\text{lin. dep.}}$}& ${\Large \substack{ \text{poly}(n,k),\\ n^{\tO(\log k)}}}$ & ${\Large \substack{ \text{poly}(n,k),\\ n^{\tO(\log k)}}}$ & $k \le O(n)$  & Allowed & Required \\ \hline
  \end{tabular}
  \caption{Comparison of our work to previous results.
      We compare runtime, sample complexity, and restrictions on the centers of the distribution: the maximum number of centers, whether linearly dependent centers are allowed, and whether the centers are required to be incoherent.
  The two subrows correspond to the cases of linearly independent and linearly dependent centers, for which we guarantee different sample complexity and runtime.}
\label{fig:comparison-fps}
  \end{center}
  \vspace{-5mm}
  \end{figure}

\paragraph{Restrictions on Input Distributions}
We detail our restrictions on the input distribution.
In the linearly independent case, if there are $k$ bias vector and $\mu$ is the incoherence of their span, and $n$ is the dimension of the samples, then we learn a product mixture in time $n^{3}$ so long as $4 \mu r < n$.
Compare this to the restriction that $\tilde\Omega(r^{7/2}\mu^5) < n$, which is the restriction of Jain and Oh--we are able to handle even a linear number of centers so long as the incoherence is not too large, while Jain and Oh can handle at most $O(n^{2/7})$ centers.
If the $k$ bias vectors are not independent, but their span has rank $r$ and  if they have maximum pairwise inner product $1-\eta$ (when scaled to unit vectors), then we learn the product mixture in time $n^{O(\log k\cdot(- \log 1 - \eta)) }$ so long as $4 \mu r \log k \cdot \log \tfrac{1}{1-\eta} < n$ (we also require a quasipolynomial number of samples in this case).

While the quasipolynomial runtime for linearly dependent vectors may not seem particularly glamorous,
we stress that the runtime depends on the separation between the vectors.
To illustrate the additional power of our result, we note that a choice of random $v_1,\ldots,v_k$ in an $r$-dimensional subspace meet this condition extremely well, as we have $\eta = 1 - \tilde O(1/\sqrt{r})$ with high probability--for, say, $k = 2r$, the algorithm of \cite{JO13} would fail in this case, since $v_1,\ldots,v_k$ are not linearly independent, but our algorithm succeeds in time $n^{O(1)}$.

This quasipolynomial time algorithm resolves an open problem of \cite{FOS08}, when restricted to distributions whose bias vectors satisfy our condition on their rank and incoherence.
We do not solve the problem in full generality, for example our algorithm fails to work when the distribution can have multiple decompositions into few centers. In such situations, the centers do not span an incoherent subspace, and thus the completion algorithms we apply fail to work. In general, the completion algorithms fail whenever the moment tensors admit many different low-rank decompositions (which can happen even when the decomposition into centers is unique, for example parity on three bits). In this case, the best algorithm we know of is the restricted brute force of Feldman, O'Donnell and Servedio.

\paragraph{Sample Complexity}

One note about sample complexity--in the linearly dependent case, we require a quasipolynomial number of samples to learn our product mixture.
That is, if there are $k$ product centers, we require $n^{\tO(\log k)}$ samples, where the tilde hides a dependence on the separation between the centers.
In contrast, Feldman, O'Donnell, and Servedio require $n^{O(k)}$ samples.
This dependence on $k$ in the sample complexity is not explicitly given in their paper, as for their algorithm to be practical they consider only constant $k$.

\paragraph{Parameter Recovery Using Tensor Decomposition}
The strategy of employing the spectral decomposition of a tensor in order to learn the parameters of an algorithm is not new, and has indeed been employed successfully in a number of settings.
In addition to the papers already mentioned which use this approach for learning product mixtures (\cite{JO14} and in some sense \cite{FOS08}, though the latter uses matrices rather than tensors),
the works of \cite{MR06,AHK12,HK13,AGHK14,BCMV14}, and many more also use this idea.
In our paper, we extend this strategy to learn a more general class of product distributions over the hypercube than could previously be tractably learned.

\subsection{Organization}
The remainder of our paper is organized as follows.
In \pref{sec:prelims}, we give definitions and background, then outline our approach to learning product mixtures over the hypercube, as well as put forth a short discussion on what kinds of restrictions we place on the bias vectors of the distribution.
In \pref{sec:tensor}, we give an algorithm for completing symmetric tensors given access only to their multilinear entries, using adversarial matrix completion as an algorithmic primitive.
In \pref{sec:pdist}, we apply our tensor completion result to learn mixtures of product distributions over the hypercube, assuming access to the precise second- and third-order moments of the distribution.
\pref{app:noisy-nnm} and \pref{app:error} contain discussions of matrix completion and learning product mixtures in the presence of sampling error, and \pref{app:whitening} contains further details about the algorithmic primitives used in learning product mixtures.

\subsection{Notation}
We use $e_i$ to denote the $i$th standard basis vector.

For a tensor $T\in\R^{n\times n\times n}$, we use $T(a,b,c)$ to denote the entry of the tensor indexed by $a,b,c\in[n]$, and we use $T(i,\cdot,\cdot)$ to denote the $i$th slice of the tensor, or the subset of entries in which the first coordinate is fixed to $i \in [n]$.
For an order-$m$ tensor $T \in \R^{n^m}$, we use $T(X)$ to represent the entry indexed by the string $X \in [n]^m$, and we use $T(Y,\cdot,\cdot)$ to denote the slice of $T$ indexed by the string $Y \in [n]^{m-2}$.
For a vector $v \in \R^n$, we use the shorthand $x^{\tensor k}$ to denote the $k$-tensor $x \tensor x \cdots \tensor x \in \R^{n \times \cdots \times n}$.

We use $\Omega \subseteq [m] \times [n]$ for the set of observed entries of the hidden matrix $M$, and $\PO$ denotes the projection onto those coordinates.

\vspace{-0.3cm}
\section{Preliminaries}
\label{sec:prelims}

In this section we present background necessary to prove our results, as well as provide a short discussion on the meaning behind the restrictions we place on the distributions we can learn. We start by defining our main problem.

\subsection{Learning Product Mixtures over the Hypercube}

A distribution $D$ over $\{\pm 1\}^n$ is called a \emph{product distribution} if every bit in a sample $x \sim D$ is independently chosen.
Let $D_1,\ldots,D_k$ be a set of product distributions over $\{\pm 1\}^n$.
Associate with each $D_i$ a vector $v_i \in[-1,1]^n$ whose $j$th entry encodes the bias of the $j$th coordinate, that is
\[
    \Pr_{x\sim\cD_i}[ x(j) = 1] = \frac{1 + v_i(j)}{2}.
\]
Define the distribution $\cD$ to be a convex combination of these product distributions, sampling
$
x \sim \cD = \{
   x \sim D_i \quad \text{with probability } w_i \}$,
where $w_i > 0$ and $\sum_{i\in[k]} w_i = 1$.
The distributions $D_1,\ldots,D_k$ are said to be the \emph{centers} of $\cD$, the vectors $v_1,\ldots,v_k$ are said to be the \emph{bias vectors}, and $w_1,\ldots,w_k$ are said to be the \emph{mixing weights} of the distribution.

\begin{problem}[Learning a Product Mixture over the Hypercube]
    Given independent samples from a distribution $\cD$ which is a mixture over $k$ centers with bias vectors $v_1,\ldots, v_k \in [-1,1]^n$ and mixing weights $w_1,\ldots,w_k > 0$, recover $v_1,\ldots,v_k$ and $w_1,\ldots,w_k$.
\end{problem}

This framework encodes many subproblems, including learning parities, a notorious problem in learning theory; the best current algorithm requires time $n^{\Omega(k)}$, and the noisy version of this problem is a standard cryptographic primitive \cite{MOS04,Feld07,Reg09, Val15}.
We do not expect to be able to learn an \emph{arbitrary} mixture over product distribution efficiently.
We obtain a polynomial-time algorithm when the bias vectors are linearly independent, and a quasi-polynomial time algorithm in the general case, though we do require an \emph{incoherence} assumption on the bias vectors (which parities do not meet), see \pref{def:incoherence}.

In \cite{FOS08}, the authors give an $n^{O(k)}$-time algorithm for the problem based on the following idea.
With great accuracy in polynomial time we may compute the pairwise moments of $\cD$,
\[
    M = \E_{x\sim \cD} [xx^T] = E_2 + \sum_{i\in[k]} w_i\cdot v_i v_i^T.
\]
The matrix $E_2$ is a diagonal matrix which corrects for the fact that $M_{jj} = 1$ always.
If we were able to learn $E_2$ and thus access $\sum_{i \in[k]} w_{i} v_i v_i^T$, the ``augmented second moment matrix,'' we may hope to use spectral information to learn $v_1,\ldots,v_k$.

The algorithm of \cite{FOS08} performs a brute-force search to learn $E_2$, leading to a runtime exponential in the rank.
By making additional assumptions on the input $\cD$ and computing higher-order moments as well, we avoid this brute force search and give a polynomial-time algorithm for product distributions with linearly independent centers:
If the bias vectors are linearly independent, a power iteration algorithm of \cite{AGHkT14} allows us to learn $\cD$ given access to both the augmented second- and third-order moments.\footnote{
    There are actually several algorithms in this space; we use the tensor-power iteration of \cite{AGHkT14} specifically.
    There is a rich body of work on tensor decomposition methods, based on simultaneous diagonalization and similar techniques (see e.g. Jenrich's algorithm \cite{H70} and \cite{LCC07}).
}
Again, sampling the third-order moments only gives access to $\E_{x\sim \cD} [x^{\otimes 3}] = E_3 + \sum_{i\in [k]} w_i \cdot v_i^{\otimes 3}$, where $E_3$ is a tensor which is nonzero only on entries of multiplicity at least two. To learn $E_2$ and $E_3$, Jain and Oh used alternating minimization and a least-squares approximation. For our improvement, we develop a tensor completion algorithm based on recursively applying the adversarial matrix completion algorithm of Hsu, Kakade and Zhang \cite{HKZ11}. In order to apply these completion algorithms, we require an \emph{incoherence} assumption on the bias vectors (which we define in the next section).

In the general case, when the bias vectors are not linearly independent, we exploit the fact that high-enough tensor powers of the bias vectors are independent, and we work with the $\tilde{O}(\log k)$th moments of $\cD$, applying our tensor completion to learn the full moment tensor, and then using \cite{AGHkT14} to find the tensor powers of the bias vectors, from which we can easily recover the vectors themselves. (the tilde hides a dependence on the separation between the bias vectors). Thus if the distribution is assumed to come from bias vectors that are \emph{incoherent} and \emph{separated}, then we can obtain a significant runtime improvement over \cite{FOS08}.

\subsection{Matrix Completion and Incoherence}
As discussed above, the matrix (and tensor) completion problem arises naturally in learning product mixtures as a way to compute the augmented moment tensors.
\begin{problem}[Matrix Completion]
    Given a set $\Omega \subseteq [m] \times [n]$ of observed entries of a hidden rank-$r$ matrix $M$, the \emph{Matrix Completion Problem} is to successfully recover the matrix $M$ given only $\PO(M)$.
\end{problem}
However, this problem is not always well-posed.
For example, consider the input matrix $M = e_1e_1^T + e_n e_n^T$.
$M$ is rank-$2$, and has only $2$ nonzero entries on the diagonal, and zeros elsewhere.
Even if we observe almost the entire matrix (and even if the observed indices are random), it is likely that every entry we see will be zero, and so we cannot hope to recover $M$.
Because of this, it is standard to ask for the input matrix to be \emph{incoherent}:
\begin{definition}
    \label{def:incoherence}
Let $U \subset \mathbb{R}^n$ be a subspace of dimension $r$.
We say that $U$ is incoherent with parameter $\mu$ if
$\max_{i \in [n]}
\|\Proj_U(e_i)\|^2
\leq
\mu
\frac{r}{n}$.
If $M$ is a matrix with left and right singular spaces $U$ and $V$, we say that $M$ is $(\mu_U,\mu_V)$-incoherent if $U$ (resp. $V$) is incoherent with parameter $\mu_U$ (resp $\mu_V$). We say that $v_1,\dots,v_k$ are incoherent with parameter $\mu$ if their span is incoherent with parameter $\mu$.
\end{definition}
Incoherence means that the singular vectors are well-spread over their coordinates.
Intuitively, this asks that every revealed entry actually gives information about the matrix.
For a discussion on what kinds of matrices are incoherent, see e.g. \cite{CR09}.
\
Once the underlying matrix is assumed to be incoherent, there are a number of possible algorithms one can apply to try and learn the remaining entries of $M$.
Much of the prior work on matrix completion has been focused on achieving recovery when the revealed entries are randomly distributed, and the goal is to minimize the number of samples needed (see e.g. \cite{CR09, R09, Getal13,H14}).
For our application, the revealed entries are not randomly distributed, but we have access to almost all of the entries ($\Omega(n^2)$ entries as opposed to the $\Omega(nr\log n)$ entries needed in the random case). Thus we use a particular kind of matrix completion theorem we call ``adversarial matrix completion,'' which can be achieved directly from the work of Hsu, Kakade and Zhang \cite{HKZ11}:
\begin{theorem}\label{thm:main_matrix_completion}
Let $M$ be an $m \times n$ rank-$r$ matrix which is $(\mu_U,\mu_V)$-incoherent, and let $\Oc \subset [m] \times [n]$ be the set of hidden indices.
If there are at most $\kappa$ elements per column and $\rho$ elements per row of $\Oc$, and if $2(\kappa\frac{\mu_U}{m} + \rho\frac{\mu_V}{n})r < 1$, then there is an algorithm that recovers $M$.
\end{theorem}
For the application of learning product mixtures, note that the moment tensors are incoherent exactly when the bias vectors are incoherent. In \prettyref{sec:tensor} we show how to apply \prettyref{thm:main_matrix_completion} recursively to perform a special type of adversarial tensor completion, which we use to recover the augmented moment tensors of $\cD$ after sampling.

Further, we note that \pref{thm:main_matrix_completion} is almost tight.
That is, there exist matrix completion instances with $\kappa/n = 1 - o(1)$, $\mu = 1$ and $r = 3$ for which finding any completion is NP-hard \cite{HMRW14,P97} (via a reduction from three-coloring), so the constant on the right-hand side is necessarily at most six.
We also note that the tradeoff between $\kappa/n$ and $\mu$ in \pref{thm:main_matrix_completion} is necessary because for a matrix of fixed rank, one can add extra rows and columns of zeros in an attempt to reduce $\kappa/n$, but this process increases $\mu$ by an identical factor.
This suggests that improving \pref{thm:main_learn_prod} by obtaining a better efficient adversarial matrix completion algorithm is not likely.

\subsection{Incoherence and Decomposition Uniqueness}
\label{sec:pinc}
In order to apply our completion techniques, we place the restriction of incoherence on the subspace spanned by the bias vectors. At first glance this may seem like a strange condition which is unnatural for probability distributions, but we try to motivate it here. When the bias vectors are incoherent and separated enough, even high-order moment-completion problems have unique solutions, and moreover that solution is equal to $\sum_{i\in[k]} w_i\cdot v_i^{\otimes m}$. In particular, this implies that the distribution must have a unique decomposition into a minimal number of well-separated centers (otherwise those different decompositions would produce different minimum-rank solutions to a moment-completion problem for high-enough order moments). Thus incoherence can be thought of as a special strengthening of the promise that the distribution has a unique minimal decomposition. Note that there are distributions which have a unique minimal decomposition but are not incoherent, such as a parity on any number of bits.

\section{Symmetric Tensor Completion from Multilinear Entries}
\label{sec:tensor}
In this section we use adversarial matrix completion as a primitive to give a completion algorithm for symmetric tensors when only a special kind of entry in the tensor is known.
Specifically, we call a string $X \in [n]^m$ \emph{multilinear} if every element of $X$ is distinct, and we will show how to complete a symmetric tensor $T \in \mathbb{R}^{n^m}$ when only given access to its multilinear entries, i.e. $T(X)$ is known if $X$ is multilinear.
In the next section, we will apply our tensor completion algorithm to learn mixtures of product distributions over the boolean hypercube.

Our approach is a simple recursion: we complete the tensor slice-by-slice, using the entries we learn from completing one slice to provide us with enough known entries to complete the next.
The following definition will be useful in precisely describing our recursive strategy:
\begin{definition}
    Define the \emph{histogram} of a string $X\in [n]^m$ to be the multiset containing the number of repetitions of each character making at least one appearance in $X$.
\end{definition}

 For example, the string $(1,1,2,3)$ and the string $(4,4,5,6)$ both have the histogram $(2,1,1)$.
    Note that the entries of the histogram of a string of length $m$ always sum to $m$, and that the length of the histogram is the number of distinct symbols in the string.

Having defined a histogram, we are now ready to describe our tensor completion algorithm.

\fbox{
\begin{minipage}{0.9\textwidth}
\begin{algorithm}[Symmetric Tensor Completion from Multilinear Moments]
    \label{alg:tensor-complete}

    {\bf Input:} The \emph{multilinear} entries of the tensor $T = \sum_{i\in[k]} w_i \cdot v_i^{\tensor m} + E$, for vectors $v_1,\ldots,v_k \in \R^n$ and scalars $w_1,\ldots,w_k>0$ and some error tensor $E$.
    {\bf Goal:} Recover the symmetric tensor $T^* = \sum_{i\in[k]} w_i\cdot v_i^{\tensor 3m}$.

\begin{enumerate}[noitemsep,topsep=1pt,partopsep=2pt]
    \item Initialize the tensor $\hat T$ with the known multilinear entries of $T$.
    \item\label{step:base}
	For each subset $Y \in [n]^{m-2}$ with no repetitions:
	\begin{itemize}[noitemsep]
	    \item Let $\hat T(Y,\cdot,\cdot)\in\R^{n\times n}$ be the tensor slice indexed by $Y$.
	    \item Remove the rows and columns of $\hat T(Y,\cdot,\cdot)$ corresponding to indices present in $Y$.
		Complete the matrix using the algorithm of \cite{HKZ11} from \pref{thm:main_matrix_completion} and add the learned entries to $\hat T$.
	\end{itemize}
    \item\label{step:larger} For $\ell = m-2,\ldots, 1$:
	\begin{enumerate}[noitemsep]
	\item For each $X \in [n]^m$ with a histogram of length $\ell$, if $\hat T(X)$ is empty:
	    \begin{itemize}[noitemsep]
		\item If there is an element $x_i$ appearing at least $3$ times,
		    let $Y = X \setminus \{x_i,x_i\}$.
		\item Else there are elements $x_i,x_j$ each appearing twice,
		    let $Y = X \setminus \{x_i,x_j\}$.
		\item Let $\hat T(Y,\cdot,\cdot)\in\R^{n\times n}$ be the tensor slice indexed by $Y$.
		\item Complete the matrix $\hat T(Y,\cdot,\cdot)$ using the algorithm from \pref{thm:main_matrix_completion} and add the learned entries to $\hat T$.
	    \end{itemize}
    \end{enumerate}
\item\label{step:symmetrize}
    Symmetrize $\hat T$ by taking each entry to be the average over entries indexed by the same subset.
    \end{enumerate}
    {\bf Output:} $\hat T$.
\end{algorithm}
\end{minipage}}
\bigskip

\begin{observation}
One might ask why we go through the effort of completing the tensor slice-by-slice, rather than simply flattening it to an $n^{m/2} \times n^{m/2}$ matrix and completing that.
The reason is that when $\Span{v_1,\dots,v_k}$ has incoherence $\mu$ and dimension $r$, $\Span{v_1^{\tensor m/2},\dots,v_k^{\tensor m/2}}$ may have incoherence as large as $\mu r^m/k$, which drastically reduces the range of parameters for which recovery is possible (for example, if $k = O(r)$ then we would need $r < n^{1/m}$).
Working slice-by-slice keeps the incoherence of the input matrices small, allowing us to complete even up to rank $r = \tilde{\Omega}(n)$.
\end{observation}

\begin{theorem}
    \label{thm:tensor-complete-alg}
    Let $T$ be a symmetric tensor of order $m$, so that $T = \sum_{i\in [k]}w_i \cdot v_i^{\tensor m}$ for some vectors $v_1,\ldots,v_k \in \R^n$ and scalars $w_1,\ldots, w_k \neq 0$.
    Let $\Span\{v_i\}$ have incoherence $\mu$ and dimension $r$.
    Given perfect access to all multilinear entries of $T$ (i.e. $E = 0$), if $4\cdot\mu\cdot r\cdot m/n < 1$, then \pref{alg:tensor-complete} returns the full tensor $T$ in time $\tilde O(n^{m+1})$.
\end{theorem}

In \pref{app:error}, we give a version of \pref{thm:tensor-complete-alg} that accounts for error $E$ in the input.

\begin{proof}
We prove that \pref{alg:tensor-complete} successfully completes all the entries of $T$ by induction on the length of the histograms of the entries.
By assumption, we are given as input every entry with a histogram of length $m$.
For an entry $X$ with a histogram of length $m-1$, exactly one
of its elements has multiplicity two, call it $x_i$, and consider the set $Y = X \setminus \{x_i,x_i\}$.
When \pref{step:base} reaches $Y$, the algorithm attempts to complete a matrix revealed from $T(Y,\cdot,\cdot) = \cP_{Y}\Paren{\sum_{i\in[k]} w_i \cdot v_i(Y) \cdot v_i v_i^T}$, where $v_i(Y) = \prod_{j \in Y} v_i(j)$,
and $\cP_{Y}$ is the projector to the matrix with the rows and columns corresponding to indices appearing in $Y$ removed.
Exactly the diagonal of $T(Y,\cdot,\cdot)$ is missing since all other entries are multilinear moments, and the $(i,i)$th entry should be $T(X)$.
Because the rank of this matrix is equal to $\dim(\Span(v_i)) = r$ and $4\mu r/n \leq 4\mu r m/n < 1$, by \pref{thm:main_matrix_completion}, we can  successfully recover the diagonal, including $T(X)$.
Thus by the end of \pref{step:base}, $\hat T$ contains every entry with a histogram of length $\ell \geq m-1$.

For the inductive step, we prove that each time \pref{step:larger} completes an iteration, $\hat T$ contains every entry with a histogram of length at least $\ell$. Let $X$ be an entry with a histogram of length $\ell$.
When \pref{step:larger} reaches $X$ in the $\ell$th iteration, if $\hat T$ does not already contain $T(X)$, the algorithm attempts to complete a matrix with entries revealed from $T(Y,\cdot,\cdot) = \sum_{i\in[k]} w_i \cdot v_i(Y) \cdot v_i v_i^T$, where $Y$ is a substring of $X$ with a histogram of the same length.
Since $Y$ has a histogram of length $\ell$, every entry of $T(Y,\cdot,\cdot)$ corresponds to an entry with a histogram of length at least $\ell + 1$, except for the $\ell \times \ell$ principal submatrix whose rows and columns correspond to elements in $Y$. Thus by the inductive hypothesis, $\hat T(Y)$ is only missing the aforementioned submatrix, and since $4 \mu r \ell/n \leq 4\mu r m/n < 1$, by \pref{thm:main_matrix_completion}, we can successfully recover this submatrix, including $T(X)$. Once all of the entries of $\hat T$ are filled in, the algorithm terminates.

Finally, we note that the runtime is $\tilde O(n^{m+1})$, because the algorithm from \pref{thm:main_matrix_completion} runs in time $\tilde{O}(n^{3})$, and we perform at most $n^{m-2}$ matrix completions because there are $n^{m-2}$ strings of length $m-2$ over the alphabet $[n]$, and we perform at most one matrix completion for each such string.
\end{proof}

\section{Learning Product Mixtures over the Hypercube}
\bigskip
\label{sec:pdist}
In this section, we apply our symmetric tensor completion algorithm (\pref{alg:tensor-complete}) to learning mixtures of product distributions over the hypercube, proving \pref{thm:main_learn_prod}.
Throughout this section we will assume exact access to moments of our input distribution, deferring finite-sample error analysis to \pref{app:error}.
We begin by introducing convenient notation.

Let $\cD$ be a mixture over $k$ centers with bias vectors $v_1,\ldots,v_k \in [-1,1]^{n}$ and mixing weights $w_1,\ldots,w_k > 0$.
Define $\cM_{m}^{\cD}\in \R^{n^m}$ to be the tensor of order-$m$ moments of the distribution $\cD$, so that $\cM_m^\cD = \E_{x\sim D} \left[x^{\tensor m}\right]$.
Define $\cT_m^{\cD} \in \R^{n^m}$ to be the symmetric tensor given by the weighted bias vectors of the distribution, so that $\cT_m^\cD = \sum_{i\in[k]} w_i\cdot v_i^{\tensor m}$.

Note that $\cT_m^{\cD}$ and $\cM_m^{\cD}$ are equal on their multilinear entries, and not necessarily equal elsewhere.
For example, when $m$ is even, entries of $\cM_m^\cD$ indexed by a single repeating character (the ``diagonal'') are always equal to 1.
Also observe that if one can sample from distribution $\cD$, then estimating $\cM_m^{\cD}$ is easy.

Suppose that the bias vectors of $\cD$ are linearly independent.
Then by \pref{thm:whitening} (due to \cite{AGHkT14}, with similar statements appearing in \cite{AHK12,HK13,AGHK14}), there is a spectral algorithm which learns $\cD$ given $\cT_2^{\cD}$ and $\cT_3^{\cD}$\footnote{We remark again that the result in \cite{AGHkT14} is quite general, and applies to a large class of probability distributions of this character. However the work deals exclusively with distributions for which $\cM_2 = \cT_2$ and $\cM_3 = \cT_3$, and assumes access to $\cT_2$ and $\cT_3$ through moment estimation.}
(we give an account of the algorithm in \pref{app:whitening}).
\begin{theorem}[Consequence of Theorem 4.3 and Lemma 5.1 in \cite{AGHkT14}]
    \label{thm:whitening}
    Let $\cD$ be a mixture over $k$ centers with bias vectors $v_1,\ldots,v_k\in[-1,1]^n$ and mixing weights $w_1,\ldots, w_k>0$.
    Suppose we are given access to $\cT_2^{\cD} = \sum_{i\in[k]} w_i \cdot v_i v_i^T$ and $\cT_3^{\cD} = \sum_{i\in[k]} w_i \cdot v_i^{\tensor 3}$.
    Then there is an algorithm which recovers the bias vectors and mixing weights of $\cD$ within $\epsilon$ in time $O(n^3 + k^4\cdot(\log \log \frac{1}{\epsilon\sqrt{w_i{\min}}}))$.
\end{theorem}

Because $\cT_2^{\cD}$ and $\cT_3^{\cD}$ are equal to $\cM_2^{\cD}$ and $\cM_3^{\cD}$ on their multilinear entries, the tensor completion algorithm of the previous section allows us to find $\cT_2^{\cD}$ and $\cT_3^{\cD}$ from $\cM_2^{\cD}$ and $\cM_3^{\cD}$ (this is only possible because $\cT_2^{\cD}$ and $\cT_3^{\cD}$ are low-rank, whereas $\cM_2^{\cD}$ and $\cM_3^{\cD}$ are high-rank).
We then learn $\cD$ by applying \pref{thm:whitening}.

A complication is that \pref{thm:whitening} only allows us to recover the parameters of $\cD$ if the bias vectors are linearly independent.
However, if the vectors $v_1,\ldots, v_k$ are not linearly independent,
we can reduce to the independent case by working instead with $v_1^{\tensor m},\ldots,v_k^{\tensor m}$ for sufficiently large $m$.
The tensor power we require depends on the \emph{separation} between the bias vectors:
\begin{definition}
    We call a set of vectors $v_1,\ldots, v_k$~$\eta$-separated if for every $i,j\in[k]$ such that $i\neq j$,
    \[
	\left|\Iprod{v_i, v_j}\right| \le \|v_i\|\cdot \|v_j\|\cdot (1-\eta).
    \]
\end{definition}

\begin{lemma}
    \label{lem:tensor-powers}
Suppose that $v_1,\ldots, v_k \in \R^n$ are vectors which are $\eta$-separated, for $\eta > 0$.
Let $m \geq \lceil\log_{\frac{1}{1-\eta}} k\rceil$.
Then $v_1^{\tensor m},\ldots, v_k^{\tensor m}$ are linearly independent.
\end{lemma}

\begin{proof}
For vectors $u,w \in \R^n$ and for an integer $t \ge 0$, we have that $\iprod{u^{\tensor t},w^{\tensor t}} = \iprod{u,w}^t$.
If $v_1,\ldots,v_k$ are $\eta$-separated, then for all $i\neq j$,
\[
    \left|\Iprod{\tfrac{v_i^{\tensor m}}{\|v_i\|^m},\tfrac{v_j^{\tensor m}}{\|v_j\|^m}}\right|  \le |(1-\eta)^m| \le \tfrac{1}{k}.
\]
Now considering the Gram matrix of the vectors $(\tfrac{v_i}{\|v_i\|})^{\tensor m}$, we have a $k \times k$ matrix with diagonal entries of value $1$ and off-diagonal entries with maximum absolute value $\tfrac{1}{k}$. This matrix is strictly diagonally dominant, and thus full rank, so the vectors must be linearly independent.
\end{proof}

\begin{remark}
We re-iterate here that in the case where $\eta = 0$, we can reduce our problem to one with fewer centers, and so our runtime is never infinite.
Specifically, if $v_i = v_j$ for some $i\neq j$, then we can describe the same distribution by omitting $v_j$ and including $v_i$ with weight $w_i + w_j$.
If $v_i = -v_j$, in the even moments we will see the center $v_i$ with weight $w_i + w_j$, and in the odd moments we will see $v_i$ with weight $w_i - w_j$.
So we simply solve the problem by taking $m' = 2m$ for the first odd $m$ so that the $v^{\tensor m}$ are linearly independent, so that both the $2m'$- and $3m'$-order moments are even to learn $w_i + w_j$ and $\pm v_i$, and then given the decomposition into centers we can extract $w_i$ and $w_j$ from the order-$m$ moments by solving a linear system.
\end{remark}

Thus, in the linearly dependent case, we may choose an appropriate power $m$, and instead apply the tensor completion algorithm to $\cM_{2m}^\cD$ and $\cM_{3m}^\cD$ to recover $\cT_{2m}^\cD$ and $\cT_{3m}^\cD$.
We will then apply \pref{thm:whitening} to the vectors $v_1^{\tensor m},\ldots, v_k^{\tensor m}$ in the same fashion.

Here we give the algorithm assuming perfect access to the moments of $\cD$ and defer discussion of the finite-sample case to \pref{app:error}.

\bigskip
\fbox{
\begin{minipage}{0.9\textwidth}
\begin{algorithm}[Learning Mixtures of Product Distributions]
    \label{alg:recovery}

    {\bf Input:} Moments of the distribution $\cD$.
    {\bf Goal:} Recover $v_1,\dots,v_k$ and $w_1,\dots,w_k$.
	\smallskip

	Let $m$ be the smallest odd integer such that $v_1^{\tensor m},\dots,v_k^{\tensor m}$ are linearly independent.
 Let $\hat M = \cM_{2m} + \hat E_2$ and $\hat T = \cM_{3m} + \hat E_3$ be approximations to the moment tensors of order $2m$ and $3m$.
    \begin{enumerate}[itemsep=0.5pt,topsep=1pt,partopsep=2pt]
    \item\label{step:tensor-comp}
	Set the non-multilinear entries of $\hat M$ and $\hat T$ to ``missing,'' and
	run \pref{alg:tensor-complete} on $\hat M$ and $\hat T$ to recover $M' = \sum_i w_i\cdot v_i^{\tensor 2m} + E_2'$ and $T' = \sum_i w_i \cdot v_i^{\tensor 3m} + E_3'$.
    \item\label{step:flattening} Flatten $M'$ to the $n^m \times n^m$ matrix $M = \sum_i w_i\cdot v_i^{\tensor m}(v_i^{\tensor m})^\top + E_2$ and similarly flatten $T'$ to the $n^m \times n^m \times n^m$ tensor $T = \sum_i w_i \cdot (v_i^{\tensor m})^{\tensor 3} + E_3$.
	\item\label{step:whitening} Run the ``whitening'' algorithm from \pref{thm:whitening} (see \pref{app:whitening}) on $(M,T)$ to recover $w_1,\dots,w_k$ and $v_1^{\tensor m},\dots,v_k^{\tensor m}$.
	\item\label{step:roots} Recover $v_1,\dots,v_k$ entry-by-entry, by taking the $m$th root of the corresponding entry in $v_1^{\tensor m},\dots,v_k^{\tensor m}$.
	\end{enumerate}
	{\bf Output:} $w_1,\dots,w_k$ and $v_1,\dots,v_k$.
\end{algorithm}
\end{minipage}}
\bigskip

Now \pref{thm:main_learn_prod} is a direct result of the correctness of \pref{alg:recovery}:

\begin{proof}[Proof of \pref{thm:main_learn_prod}]
    The proof follows immediately by combining \pref{thm:whitening} and \pref{thm:tensor-complete-alg}, and noting that the parameter $m$ is bounded by $m \leq 2+\log_{\frac{1}{1-\eta}} k$.
\end{proof}

\section*{Acknowledgements}
We would like to thank Prasad Raghavendra, Satish Rao, and Ben Recht for helpful discussions, and Moritz Hardt and Samuel B. Hopkins for helpful questions.
We also thank several anonymous reviewers for very helpful comments.

\addreferencesection
\bibliographystyle{amsalpha}
\bibliography{writeup.bib}

\appendix

\section{Tensor Completion with Noise}
\label{app:noisy-nnm}
Here we will present a version of \pref{thm:tensor-complete-alg} which account for noise in the input to the algorithm.

We will first require a matrix completion algorithm which is robust to noise.
The work of \cite{HKZ11} provides us with such an algorithm; the following theorem is a consequence of their work.\footnote{In a previous version of this paper, we derive \pref{thm:noisy_matrix_completion} as a consequence of \pref{thm:main_matrix_completion} and the work of \cite{CP09}; we refer the interested reader to \url{http://arxiv.org/abs/1506.03137v2} for the details.}

\begin{theorem}\label{thm:noisy_matrix_completion}
Let $M$ be an $m \times n$ rank-$r$ matrix which is $(\mu_U,\mu_V)$-incoherent, and let $\Oc \subset [m] \times [n]$ be the set of hidden indices.
If there are at most $\kappa$ elements per column and $\rho$ elements per row of $\Oc$, and if $2(\kappa\frac{\mu_U}{m} + \rho\frac{\mu_V}{n})r < 1$, then let $\alpha = \frac{3}{2}(\kappa\frac{\mu_U}{m} + \rho\frac{\mu_V}{n})r$ and $\beta = \frac{r}{1-\lambda}\sqrt{\frac{\kappa\rho\mu_U\mu_V}{mn}}$. In particular, $\alpha < 1$ and $\beta < 1$. Then for every $\delta > 0$, there is a semidefinite program that computes outputs $\hat{M}$ satisfying
\[\|\hat{M} - M\|_F \leq 2\delta + \frac{2\delta\sqrt{\min(n,m)}}{1-\beta}\sqrt{1+\frac{1}{1-\alpha}}. \]
\end{theorem}

We now give an analysis for the performance of our tensor completion algorithm, \pref{alg:tensor-complete}, in the presence of noise in the input moments.
This will enable us to use the algorithm on empirically estimated moments.
\begin{theorem} \label{thm:completion-err}
    Let $T^*$ be a symmetric tensor of order $m$, so that $T^* = \sum_{i\in [k]}w_i \cdot v_i^{\tensor m}$ for some vectors $v_1,\ldots,v_k \in \R^n$ and scalars $w_1,\ldots, w_k \neq 0$.
    Let $\Span\{v_i\}$ have incoherence $\mu$ and dimension $r$.
    Suppose we are given access to $T = T^* + E$, where $E$ is a noise tensor with $|E(Y)| \le \epsilon$ for every $Y \in [n]^m$.
    Then if
    \[
	4 \cdot k \cdot \mu \cdot m \le n,
    \]
    Then \pref{alg:tensor-complete} recovers a symmetric tensor $\hat T$ such that
    \[
	\| \hat T(X,\cdot,\cdot) - T^*(X,\cdot,\cdot) \|_F \le 4\cdot \epsilon \cdot (5n^{3/2})^{m-1},
    \]
    for any slice $T(X,\cdot,\cdot)$ indexed by a string $X \in [n]^{m-2}$,
    in time $\tilde O(n^{m+1})$. In particular, the total Frobenius norm error $\|\hat T - T^*\|_F$ is bounded by $4\cdot \epsilon \cdot (5n^{3/2})^{\frac{3}{2}m - 2}$.
\end{theorem}

\begin{proof}
    We proceed by induction on the histogram length of the entries: we will prove that an entry with a histogram of length $\ell$ has error at most $\epsilon (5n^{3/2})^{m - \ell}$.

    In the base case of $\ell = m$, we have that by assumption, every entry of $E$ is bounded by $\epsilon$.

    Now, for the inductive step, consider an entry $X$ with a histogram of length $\ell \le m - 1$.
    In filling in the entry $T(X)$, we only use information from entries with shorter histograms, which by the inductive hypothesis each have error at most $\alpha = \epsilon (5n^{3/2})^{m-\ell-1}$.
    Summing over the squared errors of the individual entries, the squared Frobenius norm error of
    the known entries in the slice in which $T(X)$ was completed, pre-completion is at most $n^2\alpha^2$.
    Due to the assumptions on $k,\mu,m,n$, by \pref{thm:noisy_matrix_completion}, matrix completion amplifies the Frobenius norm error of $\beta$ to at most a Frobenius norm error of $5\beta \cdot n^{1/2}$.
    Thus, we have that the Frobenius norm of the slice $T(X)$ was completed in, post-completion, is at most $5 n^{3/2} \alpha$, and therefore that the error in the entry $T(X)$ is as most $\epsilon \cdot (5n^{3/2})^{m-\ell}$, as desired.

    This concludes the induction.
    Finally, as our error bound is per entry, it is not increased by the symmetrization in \pref{step:symmetrize}.
    Any slice has at most one entry with a histogram of length one, $2n-2$ entries with a histogram of length two, and $n^2-(2n-1)$ entries with a histogram of length three. Thus the total error in a slice is at most $4\cdot \epsilon \cdot (5n^{3/2})^{m-1}$, and there are $n^{m-2}$ slices.
\end{proof}

\section{Empirical Moment Estimation for Learning Product Mixtures}\label{app:error}

In \pref{sec:pdist}, we detailed our algorithm for learning mixtures of product distributions while assuming access to exact moments of the distribution $\cD$.
Here, we will give an analysis which accounts for the errors introduced by empirical moment estimation.
We note that we made no effort to optimize the sample complexity, and that a tighter analysis of the error propagation may well be possible.

\smallskip

\fbox{
\begin{minipage}{0.9\textwidth}
    \begin{algorithm}[Learning product mixture over separated centers]
    \label{alg:learning_dep}

    {\bf Input:} $N$ independent samples $x_1,\ldots, x_N$ from $\cD$, where $\cD$ has bias vectors with separation $\eta > 0$.
    {\bf Goal:} Recover the bias vectors and mixing weights of $\cD$.
\smallskip

	Let $m$ be the smallest odd integer for which $v_1^{\tensor m},\ldots,v_k^{\tensor m}$ become linearly independent.
    \begin{enumerate}
	\item Empirically estimate $\cM_{2m}^{\cD}$ and $\cM_{3m}^{\cD}$ by calculation $M := \frac{1}{N}\sum_{i \in [N]} (x_i^{\tensor m})(x_i^{\tensor m})^\top$ and $T := \frac{1}{N} \sum_{i\in[N]} (x_i^{\tensor m})^{\tensor 3}$.
	    \label{step:estimation}
	\item Run \pref{alg:recovery} on $M$ and $T$.
	    \label{step:nnm-big-err}
    \end{enumerate}
    {\bf Output:} The approximate mixing weights $\hat w_1,\ldots, \hat w_k$, and the approximate vectors $\hat v_1,\ldots, \hat v_k$.
\end{algorithm}
\end{minipage}}
\bigskip

\bigskip
\begin{theorem}[\pref{thm:main_learn_prod} with empirical moment estimation]
    \label{thm:learn-big}
    Let $\cD$ be a product mixture over $k$ centers with bias vectors $v_1,\ldots,v_k \in [-1,1]^{n}$ and mixing weights $w_1,\ldots,w_k>0$.
    Let $m$ be the smallest odd integer for which $v_1^{\tensor m},\ldots,v_k^{\tensor m}$ are linearly independent (if $v_1,\ldots,v_k$ are $\eta$-separated for $\eta > 0$, then $m \le \log_{\frac{1}{1-\eta}} k$).
Define $M_{2m} = \sum_{i\in[k]} v_i^{\tensor m} (v_i^{\tensor m})^\top$.
    Suppose
        \[
	    4\cdot m \cdot r \cdot \mu \le n,
        \]
    where $\mu$ and $r$ are the incoherence and dimension of the space $\Span\{v_i\}$ respectively.
    Furthermore, let $\beta \leq\min\left( O(1/k\sqrt{w_{\max}}),\tfrac{1}{40}\right)$ be suitably small, and let the parameter $N$ in \pref{alg:learning_dep} satisfy $N \geq \frac{2}{\epsilon^2}(4\log n + \log \frac{1}{\delta})$ for $\epsilon$ satisfying
    \[
    \epsilon \leq \frac{\beta \cdot \sigma_k(M_{2m})}{4\cdot (5n^{3/2})^{3m-2}}\min\left(\frac{1}{6\sqrt{w_{\max}}},\frac{\sigma_k(M)^{1/2}}{(5n^{3/2})^{3m/2}}\right)
    \]

    Finally, pick any $\eta \in (0,1)$.
    Then with probability at least $1-\delta-\eta$, \pref{alg:learning_dep} returns vectors $\hat v_1,\dots, \hat v_k$ and mixing weights $\hat w_1,\dots,\hat w_k$ such that
    \[
	\|\hat v_i - v_i\| \le \sqrt{n} \cdot \left(10\cdot\beta + 60\cdot\beta\cdot\|M_{2m}\|^{1/2} + \frac{\beta \cdot \sigma_k(M_{2m})}{6\sqrt{w_{\max}}}\right)^{1/m}, \quad \text{and} \quad |\hat w_i - w_i| \le 40 \beta,
    \]
    and runs in time $n^{O(m)}\cdot O(N \cdot \poly(k)\log(1/\eta) \cdot(\log k + \log\log (\tfrac{w_{\max}}{\epsilon})))$.
    In particular, a choice of $N \ge n^{\tilde O(m)}$ gives sub-constant error, where the tilde hides the dependence on $w_{\min}$ and $\sigma_{k}(M_{2m})$.
\end{theorem}

Before proving \pref{thm:learn-big}, we will state  state the guarantees of the whitening algorithm of \cite{AGHkT14} on noisy inputs, which is used as a black box in \pref{alg:recovery}.
We have somewhat modified the statement in \cite{AGHkT14} for convenience;
for a breif account of their algorithm, as well as an account of our modifications to the results as stated in \cite{AGHkT14}, we refer the reader to \pref{app:whitening}.
\begin{theorem}[Corollary of Theorem 4.3 in \cite{AGHkT14}]
    \label{thm:whitening-err}
    Let $v_1,\ldots,v_k\in[-1,1]^n$ be vectors and let $w_1,\ldots,w_k>0$ be weights.
    Define
$M = \sum_{i\in[k]} w_i \cdot v_i v_i^T$ and $T = \sum_{i\in[k]} w_i \cdot v_i^{\tensor 3}$, and
    suppose we are given  $\hat M = M + E_M$ and $\hat T = T + E_T$, where $E_M\in \R^{n \times n}$ and $E_T \in \R^{n\times n\times n}$ are symmetric  error terms such that
    \[
	2\beta \coloneqq \frac{6\|E_M\|_F\sqrt{w_{\max}}}{\sigma_k(M)} + \frac{\sqrt{k}\|E_T\|_F}{\sigma_k(M)^{3/2}} < O\Paren{\frac{1}{\sqrt{w_{\max}}\cdot k}}.
\]
    Then there is an algorithm that recovers vectors $\hat v_1,\ldots,\hat v_k$ and weights $\hat w_1,\ldots,\hat w_k$ such that for all $i\in[n]$,
    \[
	\| v_i - \hat v_i\| \le  \|E_M\|^{1/2} + 60 \|M\|^{1/2} \beta + 10 \beta, \quad \text{and} \quad \left|w_i - \hat w_i \right| \le 40 \beta,
    \]
    with probability $1-\eta$ in time $O(L \cdot k^3\cdot(\log k + \log\log (\tfrac{1}{\sqrt{w_{\max}}\cdot\epsilon})))$, where $L$ is $\poly(k)\log(1/\eta)$.
\end{theorem}

Having stated the guarantees of the whitening algorithm, we are ready to prove \pref{thm:learn-big}.
\begin{proof}[Proof of \pref{thm:learn-big}]
    We account for the noise amplification in each step.

\smallskip

    {\bf Step \ref{step:estimation}:}
    In this step, we empirically estimate the multilinear moments of the distribution.
    We will apply concentration inequalities on each entry individually.
    By a Hoeffding bound, each entry concentrates within $\epsilon$ of its expectation with probability $1 - \exp(-\tfrac{1}{2} N \cdot \epsilon^2)$. Taking a union bound over the $\binom{n}{2m} + \binom{n}{3m}$ moments we must estimate, we conclude that with probability at least $1 - \exp(-\tfrac{1}{2}N \cdot \epsilon^2 + 4m \log n)$, all moments concentrate to within $\epsilon$ of their expectation.
    Setting $N = \tfrac{2}{\epsilon^2}(4m \log n + \log \tfrac{1}{\delta})$, we have that with probability $1-\delta$, every entry concentrates to within $\epsilon$ of its expectation.

    \bigskip
    Now, we run \pref{alg:recovery} on the estimated moments.
    \bigskip

    {\bf Step \ref{step:tensor-comp} of \pref{alg:recovery}:}
    Applying \pref{thm:completion-err}, we see that the error satisfies $\|E'_2\|_F \leq 4\cdot \epsilon \cdot (5n^{3/2})^{3m-2}$ and $\|E'_3\|_F \leq 4\cdot \epsilon \cdot (5n^{3/2})^{\frac{9}{2}m-2}$.

    \bigskip
    {\bf Step \ref{step:flattening} of \pref{alg:recovery}:} No error is introduced in this step.
    \bigskip

    {\bf Step \ref{step:whitening} of \pref{alg:recovery}:}
    Here, we apply \pref{thm:whitening-err} out of the box, where our vectors are now the $v_i^{\tensor m}$.
    The desired result now follows immediately for the estimated mixing weights, and for the estimated tensored vectors we have $\|u_i - v_i^{\tensor m} \| \le 10 \cdot\beta + 60\cdot \beta \|M\|^{1/2} + \|E_2'\|$, for $\beta$ as defined in \pref{thm:learn-big}.
    Note that $\|E_2'\| \leq \|E_2'\|_F \leq \beta\cdot \sigma_k(M_{2m})/6\sqrt{w_{\max}}$, so
    let $\gamma = 10 \cdot\beta + 60\cdot \beta \|M\|^{1/2} + \frac{\beta \cdot \sigma_k(M_{2m})}{6\sqrt{w_{\max}}}$.

    \bigskip

    {\bf Step \ref{step:roots} of \pref{alg:recovery}:}
    Let $u_i^*$ be the restriction of $u_i$ to the single-index entries, and let $v_i^*$ be the same restriction for $v_i^{\tensor m}$.
    The bound on the error of the $u_i$ applies to restrictions,
    so we have
    $
	\|v_i^* - u_i^*\| \le \gamma.
    $
    So the error in each entry is bounded by $\gamma$.
    By the concavity of the $m$th root, we thus have that $\|v_i - \hat v_i\| \le \sqrt{n} \cdot \gamma^{1/m}$.
    \bigskip

    To see that choosing $N \ge n^{\tilde O(m)}$ gives sub-constant error, calculations suffice; we only add that $\|M_{2m}\| \le r n^{m}$, where we have applied a bound on the Frobenius norm of $\|M_{2m}\|$.
    The tilde hides the dependence on $w_{\min}$ and $\sigma_{k}(M_{2m})$.
    This concludes the proof.
\end{proof}

\section{Recovering Distributions from Second- and Third-Order Tensors}
\label{app:whitening}

In this appendix, we give an account of the algorithm of \cite{AGHkT14} which, given access to estimates of $M_V^{\cD}$ and $T_V^{\cD}$, can recover the parameters of $\cD$.
We note that the technique is very similar to those of \cite{AHK12,HK13,AGHK14}, but we use the particular algorithm of \cite{AGHkT14}.
In previous sections, we have given a statement that follows from their results; here we will detail the connection.

In \cite{AGHkT14}, the authors show that for a family of distributions with parameters $v_1,\ldots,v_k \in \R^n$ and $w_1,\ldots,w_k > 0$, if the $v_1,\ldots, v_k$ are linearly independent and one has approximate access to $M_V := \sum_{i \in[k]} w_i v_i v_i^T$ and $T_V := \sum_{i\in[k]} w_i \cdot v_i ^{\tensor 3}$, then the parameters can be recovered.
For this, they use two algorithmic primitives: singular value decompositions and tensor power iteration.

Tensor power iteration is a generalization of the power iteration technique for finding matrix eigenvectors to the tensor setting (see e.g. \cite{AGHkT14}).
The generalization is not complete, and the convergence criteria for the method are quite delicate and not completely understood, although there has been much progress in this area of late (\cite{AGJ14,AGJ14b,GHJY15}).
However, it is well-known that when the input tensor $T \in \R^{n\times n \times n}$ is decomposable into $k < n$ symmetric orthogonal rank-1 tensors, i.e. $T = \sum_{i \in [k]} v_i^{\tensor 3}$ where $k < n$ and $\iprod{v_i,v_j} = 0$ for $i \neq j$, then it is possible to recover $v_1,\ldots, v_k$ using tensor power iteration.

The authors of \cite{AGHkT14} prove that this process is robust to some noising of $T$:
\begin{theorem}[Theorem 5.1 in \cite{AGHkT14}]
    \label{thm:tpi}
    Let $\tilde T = T + E \in \R^{k \times k \times k}$ be a symmetric tensor, where $T$ has the decomposition $T = \sum_{i\in[k]} \lambda_i \cdot u_i \tensor u_i \tensor u_i$ for orthonormal vectors $u_1,\cdots, u_k$ and $\lambda_1,\ldots, \lambda_k >0$, and $E$ is a tensor such that $\|E\|_F \le \beta$.
    Then there exist universal constants $C_1,C_2,C_3 > 0$ such that the following holds.
    Choose $\eta \in (0,1)$, and suppose
    \[
	\beta \le C_1 \cdot \frac{\lambda_{\min}}{k}
    \]
    and also
    \[
	\sqrt{\frac{\ln(L/\log_2(k/\eta))}{\ln k}}\cdot
	\left(1 - \frac{\ln(\ln(L/\log_2(k/\eta))) + C_3}{4\ln(L/\log_2(k/\eta))} - \sqrt{\frac{\ln 8}{\ln(L/\log_2(k/\eta))}}\right) \ge 1.02\left(1 + \sqrt{\frac{\ln 4}{\ln k}}\right).
    \]
    Then there is a tensor power iteration based algorithm that recovers vectors $\hat u_1,\ldots,\hat u_k$ and coefficients $\hat\lambda_1,\ldots,\hat\lambda_k$ with probability at least $1-\eta$ such that for all $i\in[n]$,
    \[
	\|\hat u_i - u_i\| \le \beta \frac{8}{\lambda_i}, \quad \text{and} \quad |\hat\lambda_i - \lambda_i| \le 5 \beta,
    \]
    in $O(L \cdot k^3\cdot(\log k + \log\log (\tfrac{\lambda_{\max}}{\beta})))$ time.
    The conditions are met when $L = \poly(k)\log(1/\eta)$.
\end{theorem}

The idea is then to take the matrix $M_V$, and apply a \emph{whitening} map $W = (M_V^{\dagger})^{1/2}$ to orthogonalize the vectors.
Because $v_1,\ldots,v_k$ are assumed to be linearly independent, and because
$W M W = \sum_{i \in [k]} w_i (Wv_i) (Wv_i)^T= \text{Id}_k$, it follows that the $\sqrt{w_i}\cdot Wv_i$ are orthogonal vectors.
Now, applying the map $W \in \R^{k \times n}$ to every slice of $T$ in every direction, we obtain a new tensor $T_W = \sum_{i \in[k]} w_i (Wv_i)^{ \tensor 3}$, by computing each entry:
\[
    T(W,W,W)_{a,b,c} \coloneqq T_W(a,b,c) = \sum_{1\le a',b',c' \le n} W^T(a',a)\cdot W^T(b',b)\cdot W^T(c',c)\cdot T(a',b',c').
\]
From here on out we will use $T(A,A,A)$ to denote this operation on tensors.
The tensor $T_W$ thus has an orthogonal decomposition.
Letting $u_i = \sqrt{w_i} W v_i$, we have that $T = \sum_{i\in[k]} \tfrac{1}{\sqrt{w_i}} \cdot u_i^{\tensor 3}$.
Applying tensor power iteration allows the recovery of the $u_i = \sqrt{w_i} \cdot Wv_i$ and the weights $\tfrac{1}{\sqrt{w_i}}$, from which the $v_i$ are recoverable.

The theorem \pref{thm:whitening-err} is actually the consequence of \pref{thm:tpi} and the following proposition, which controls the error propagation in the whitening step.
\begin{proposition}[Consequence of Lemma 12 of \cite{HK13}]
    \label{prop:white-err}
    Let $M_2 = \sum_{i\in[k]} \lambda_i\cdot u_i u_i^T$ be a rank-$k$ PSD matrix, and let $\hat M$ be a symmetric matrix whose top $k$ eigenvalues are positive.
    Let $T = \sum_{i\in[k]} \lambda_i\cdot u_i^{\tensor 3}$, and let $\hat T = T + E_T$ where $E_T$ is a symmetric tensor with $\|E_T\|_F \le \gamma$.

    Suppose $\|M_2 - \hat M\|_F \le \epsilon \sigma_k(M_2)$, where $\sigma_k(M)$ is the $k$th eigenvalue of $M_2$.
    Let $U$ be the square root of the pseudoinverse of $M_2$, and let $\hat U$ be the square root of the pseudoinverse of the projection of $\hat M$ to its top $k$ eigenvectors.
    Then
    \[
	\|T(U, U, U) - \hat T(\hat U,\hat U,\hat U)\| \le \frac{6}{\sqrt{\lambda_{\min}}} \epsilon + \gamma \cdot \|\hat U\|^2\|\hat U\|_F
    \]
\end{proposition}
\begin{proof}
    We use the following fact, which is given as Lemma 12 in \cite{HK13}.
    \[
	\|T(U, U, U) - \hat T(\hat U,\hat U,\hat U)\| \le \frac{6}{\sqrt{\lambda_{\min}}} \epsilon + \|E_T(\hat U,\hat U, \hat U)\|_2.
    \]
The proof of this fact is straightforward, but requires a great deal of bookkeeping; we refer the reader to \cite{HK13}.

It remains to bound $\|E_T(\hat U,\hat U, \hat U)\|_2$. Some straightforward calculations yield the desired bound,
    \begin{align*}
	\|E(\hat U, \hat U, \hat U)\|_2
	&\le \sum_{i} \|(\hat U e_i) \tensor \hat U^T E_i \hat U\|
	\le \sum_{i} \|\hat U e_i\|_2 \| \hat U^T E_i \hat U\|\\
	&\le \|\hat U\|^2 \cdot \sum_{i} \|\hat U e_i\|_2 \| E_i \|
	\le \|\hat U\|^2 \cdot \sum_{i} \|\hat U e_i\|_2 \| E_i \|_F\\
	&\le \|\hat U\|^2 \cdot \sqrt{\sum_{i} \|\hat U e_i\|^2_2}\sqrt{\sum_i \| E_i \|^2_F}
	\le \|\hat U\|^2\cdot \|\hat U\|_F\cdot \|E\|_F,
    \end{align*}
    where we have applied the triangle inequality, the behavior of the spectral norm under tensoring, the submultiplicativity of the norm, and Cauchy-Schwarz.
\end{proof}

We now prove \prettyref{thm:whitening-err}.
\begin{proof}[Proof of \pref{thm:whitening-err}]
    Let $\hat U$ be the square root of the projection of $\hat M$ to its top $k$ eigenvectors.
 Note that $\|\hat U\| \leq \sigma_k(M_2)^{-1/2}$, $\|\hat U\|_F \leq \sqrt{k}\sigma_k(M_2)^{-1/2}$, and thus by \prettyref{prop:white-err}, the error $E$ in \prettyref{thm:tpi} satisfies
\[
2\beta \coloneqq \|E\|_F \leq \frac{6\|E_M\|_F}{\sigma_k(M_2) \sqrt{\lambda_{\min}}} + \frac{\|E_T\|_F\sqrt{k}}{\sigma_k(M_2)^{3/2}}.
\]

    Suppose $1/40 \ge 2\beta\ge \|E\|_F$.
    Applying \pref{prop:white-err}, we obtain vectors $u_1,\ldots,u_k$ and scaling factors $\lambda_1,\ldots,\lambda_k$ such that $\|u_i - \sqrt{w_i} \cdot M^{-1/2} v_i\| \le 16 \cdot \beta \cdot \sqrt{w_i}$ and $| \frac{1}{\sqrt{w_i}} - \lambda_i | \le 5\cdot \beta$.
    The $w_i$ are now recovered by taking the inverse square of the $\lambda_i$, so we have that when $10\beta < \tfrac{1}{4} \le \tfrac{1}{4} \lambda_i$,
    \[
	\left|\hat w_i -w_i\right|
	= \left|\frac{1}{\lambda_i^2} - w_i\right|
	\le \left|\frac{1}{\lambda_i^2} - \frac{1}{(\lambda_i \pm 10\beta)^2}\right|
	\le 5\beta \cdot \frac{2 \lambda_i - 10\beta}{\lambda_i^2(\lambda_i - 10\beta)^2} \le 40 \beta,
    \]
    where to obtain the second inequality we have taken a Taylor expansion, and in the final inequality we have used the fact that $10\beta < \tfrac{1}{4} \lambda_i$.

    We now recover $v_i$ by taking $\hat v_i = \lambda_i\cdot \hat U u_i$, so we have
    \begin{align*}
	\|\hat v_i - v_i\|
	&\le \|\lambda_i\cdot \hat U \sqrt{w_i}\cdot M^{-1/2} v_i  - v_i\|
	+\|\lambda_i\cdot \hat U (u_i - \sqrt{w_i} M^{-1/2} v_i)\|\\
	& \le (\lambda_i \cdot \sqrt{w_i}) \| (\hat U\cdot M^{-1/2}-I)v_i\| + (1-\lambda_i\sqrt{w_i}) \|v_i\|
	+\|\hat U\| \cdot 16 \beta \lambda_i \sqrt{w_i} \\
	& \le (1+10\beta) \| \hat U\cdot M^{-1/2} - I\| + 10\beta
	+\|\hat U\| \cdot 16\beta(1+10\beta) \\
	& \le (1+10\beta) \| \hat U\cdot M^{-1/2} - I\| + 10\beta
	+ \|\hat U\| \cdot 16\beta(1+10\beta).
    \end{align*}
    It now suffices to bound $\|\hat U M^{-1/2} - I \|$, for which it in turn suffices to bound $\|M^{-1/2} \hat U\hat U M^{-1/2} - I\|$, since the eigenvalues of $AA^T$ are the square eigenvalues of $A$.
    Consider $\| (M^{-1/2}\Pi_k (M + E_M)\Pi_k) M^{-1/2} - I \|$, where $\Pi_k$ is the projector to the top $k$ eigenvectors of $M$.
    Because both matrices are PSD, finally this reduces to bounding $\| M - \Pi_k(M+ E_M) \Pi_k\|$.
    Since $M$ is rank $k$, we have that $\| M - \Pi_k(M+ E) \Pi_k\| = \sigma_{k+1}(E_M) \leq \|E_M\|$.

    Thus, taking loose bounds, we have
    \[
	\|v_i - \hat v_i\|
	\le \|E_M\|^{1/2} + 60 \beta \cdot \|M_2\|^{1/2} + 10\beta,
    \]
    as desired.
\end{proof}

\end{document}
